\documentclass[aps, pra, reprint, superscriptaddress]{revtex4-1}

\pdfoutput=1

\usepackage{header}

\graphicspath{{graphics/}}

\begin{document}

\bibliographystyle{apsrev4-1}

\title{Approximation Algorithms for Complex-Valued Ising Models on Bounded Degree Graphs}

\author{Ryan L. Mann}
\email{mail@ryanmann.org}
\homepage{http://www.ryanmann.org}
\affiliation{Centre for Quantum Computation and Communication Technology, \\ Centre for Quantum Software and Information, \\ Faculty of Engineering \& Information Technology, University of Technology Sydney, NSW 2007, Australia}

\author{Michael J. Bremner}
\affiliation{Centre for Quantum Computation and Communication Technology, \\ Centre for Quantum Software and Information, \\ Faculty of Engineering \& Information Technology, University of Technology Sydney, NSW 2007, Australia}

\begin{abstract}
We study the problem of approximating the Ising model partition function with complex parameters on bounded degree graphs. We establish a deterministic polynomial-time approximation scheme for the partition function when the interactions and external fields are absolutely bounded close to zero. Furthermore, we prove that for this class of Ising models the partition function does not vanish. Our algorithm is based on an approach due to Barvinok for approximating evaluations of a polynomial based on the location of the complex zeros and a technique due to Patel and Regts for efficiently computing the leading coefficients of graph polynomials on bounded degree graphs. Finally, we show how our algorithm can be extended to approximate certain output probability amplitudes of quantum circuits.
\end{abstract}

\maketitle

\section{Introduction}
\label{section:Introduction} 

The Ising model partition function plays an important role in combinatorics and statistical physics. In this paper we study the problem of approximating the Ising model partition function in the complex parameter regime on bounded degree graphs. This work is motivated by the close relationship to quantum computation, where it can be shown that approximate evaluations of these partition functions can encode arbitrary quantum computations~\cite{de2011quantum}. A classic result of Jaeger, Vertigan, and Welsh~\cite{jaeger1990computational} showed that exactly evaluating these partition functions is \mbox{\textsc{\#P}-hard}. This was shown to remain true in the approximate case~\cite{goldberg2017complexity} and when restricted to graphs of bounded degree~\cite{fujii2017commuting}. Therefore, it seems unlikely that an efficient algorithm exists for approximating the partition function for general parameters on bounded degree graphs. Furthermore, it has been conjectured that this problem remains hard on average over certain classes of interactions and external fields~\cite{gao2017quantum, boixo2018characterizing, bermejo2018architectures}. Resolving these conjectures is important for understanding the complexity of quantum computing.

We establish a deterministic polynomial-time approximation scheme for the Ising model partition function on bounded degree graphs when the interactions and external fields are absolutely bounded close to zero \mbox{(Corollary~\ref{corollary:IsingModelPartitionFunctionBoundedDegreeGraphs})}. This provides a lower bound on when the interactions and external fields cause approximations of the Ising model partition function to transition from being contained in \textsc{P} to being \mbox{\textsc{\#P}-hard}. Our algorithm is based on an approach due to Barvinok~\cite{barvinok2015computing, barvinok2016computing, barvinok2016combinatorics} for approximating evaluations of a polynomial based on the location of the complex zeros and a technique due to Patel and Regts~\cite{patel2017deterministic} for efficiently computing the leading coefficients of graph polynomials on bounded degree graphs.

Barvinok's approach considers the Taylor expansion of the logarithm of a polynomial about an easy to evaluate point. Suppose that we can show that the complex zeros of the polynomial lie in the exterior of a closed disc centred at this point, then it follows that a truncated Taylor expansion provides an additive approximation to the logarithm of the polynomial at any point in the interior of this closed disc. Now observe that an additive approximation of the logarithm of a polynomial corresponds to a multiplicative approximation of the polynomial.

To construct an algorithm from this approach we need to be able to compute the coefficients of the truncated Taylor expansion. Barvinok showed that computing these coefficients can be reduced to computing the leading coefficients of the polynomial itself. However, to achieve the accuracy required for an approximation scheme, we require a number of leading coefficients that is logarithmic in the degree of the polynomial. For many combinatorial structures, directly computing these coefficients requires quasi-polynomial time. 

Patel and Regts~\cite{patel2017deterministic} showed that, for several classes of graph polynomials on bounded degree graphs, the leading coefficients can be computed in polynomial time. Their approach is based on expressing the coefficients as linear combinations of connected induced subgraph counts of size logarithmic in the size of the graph. It then follows from a result due to Borgs et al.~\cite{borgs2013left}, which states that, for bounded degree graphs, we can efficiently enumerate all connected induced subgraphs of logarithmic size.

Barvinok and Sober\'on~\cite{barvinok2017computing} established a deterministic quasi-polynomial time algorithm for approximating the multivariate graph homomorphism partition function on bounded degree graphs when the matrix entries are absolutely bounded close to one. In the case that all matrix entries are exactly equal to one, the partition function is easy to evaluate. Barvinok and Sober\'on proved that for bounded degree graphs when the matrix entries are absolutely bounded close to one, the partition function does not vanish. Finally, they proved that the leading coefficients can be computed in quasi-polynomial time. Patel and Regts~\cite{patel2017deterministic} improved this to a deterministic polynomial-time algorithm by showing that the coefficients can be expressed as linear combinations of connected induced subgraph counts. 

In order to establish a polynomial-time approximation scheme for the Ising model partition function, we provide an approximation-preserving polynomial-time reduction to a restricted version of the multivariate graph homomorphism partition function \mbox{(Proposition~\ref{proposition:IsingModelGraphHomomorphismPartitionFunctionReduction})}. We extend the results of Barvinok and Sober\'on~\cite{barvinok2017computing} and Patel and Regts~\cite{patel2017deterministic} to give an algorithm for approximating this restricted version of the multivariate graph homomorphism partition function on bounded degree graphs when the matrix entries are absolutely bounded close to one \mbox{(Theorem~\ref{theorem:GraphHomomorphismPartitionFunctionBoundedDegreeGraphs})}. As a consequence, we obtain a deterministic polynomial-time approximation scheme for the Ising model partition function on bounded degree graphs when the interactions and external fields are absolutely bounded sufficiently close to zero. Furthermore, we prove that in this case the Ising model partition function does not vanish \mbox{(Corollary~\ref{corollary:IsingModelPartitionFunctionZeroFreeRegion})}. This may be of independent interest in statistical physics as the possible points of physical phase transitions are exactly the real limit points of complex zeros~\cite{sokal2005multivariate}.

Previous work by Liu, Sinclair, and Srivastava~\cite{liu2019ising} studied the problem of approximating the ferromagnetic Ising model partition function based on the location of complex zeros. They gave a deterministic polynomial-time approximation scheme for the Ising model partition function in the ferromagnetic regime for all complex external fields that are not purely imaginary. This can be seen as an algorithmic consequence of the classic Lee-Yang Theorem~\cite{lee1952statistical}, which states that the ferromagnetic Ising model partition function does not vanish except when the external fields are purely imaginary. Peters and Regts~\cite{peters2018location} generalised this result by determining the exact location of zeros in the ferromagnetic and anti-ferromagnetic regime as a function of the inverse temperature and the maximum degree.

Further work has considered the problem of approximating the Ising model partition function on bounded degree graphs based on the decay of correlations property. Sinclair, Srivastava, and Thurley~\cite{sinclair2014approximation} established a deterministic polynomial-time approximation scheme for the anti-ferromagnetic Ising model partition function on graphs of maximum degree at most $\Delta$ when the interactions and external fields lie in the uniqueness region of the Gibbs measure on the infinite \mbox{$\Delta$-regular} tree, which is exactly the region that the decay of correlation property holds. Sly and Sun~\cite{sly2012computational} showed that for interactions outside of this region, unless \mbox{\textsc{RP}=\textsc{NP}}, there is no fully polynomial-time randomised approximation scheme for the anti-ferromagnetic Ising model partition function on graphs of maximum degree at most $\Delta\geq3$. Independent work by Galanis, \v{S}tefankovi\v{c}, and Vigoda~\cite{galanis2016inapproximability} established a similar result in the case of no external field. Liu, Sinclair, and Srivastava~\cite{liu2018fisher} showed that, in the case of no external field, the Ising model partition function has no zeros in a complex neighbourhood of the decay of correlation regime. This establishes a formal relationship between these two approaches.

Our final result is a polynomial-time algorithm for approximating certain output probability amplitudes of quantum circuits \mbox{(Corollary~\ref{corollary:IQPBoundedDegreeGraphs})}. This algorithm is based on the observation that complex-valued Ising model partition functions arise in the output probability amplitudes of quantum circuits~\cite{de2011quantum, iblisdir2014low}. We focus on a class of commuting quantum circuits, known as Instantaneous Quantum Polynomial-time (IQP) circuits~\cite{shepherd2009temporally}, where the mapping to the Ising model partition function is the most straightforward~\cite{shepherd2009temporally, shepherd2010binary, fujii2017commuting}. Bremner, Montanaro, and Shepherd~\cite{bremner2016average} showed that general IQP circuits cannot be efficiently classically simulated under the assumption that the Polynomial Hierarchy does not collapse and the Ising model partition function is \mbox{\textsc{\#P}-hard} on average over certain classes of interactions and external fields. Furthermore, IQP circuits are known to become universal for quantum computation under postselection~\cite{bremner2010classical}, therefore, approximating output probability amplitudes of IQP circuits is equivalent to approximating output probability amplitudes of universal quantum circuits. Our algorithm allows us to approximate a certain output probability amplitude of a quantum circuit when the corresponding graph has bounded degree and the interactions and external fields are absolutely bounded close to zero. Eldar and Mehraban~\cite{eldar2018approximating} used a similar approach to derive a quasi-polynomial time algorithm for approximating the permanent of a random matrix with unit variance and vanishing mean in the context of linear optical quantum computing.

This paper is structured as follows. In \mbox{Section~\ref{section:GraphHomomorphismPartitionFunctions}}, we introduce the multivariate graph homomorphism partition function and establish a deterministic polynomial-time algorithm for approximating a restricted version of this partition function on bounded degree graphs when the matrix entries are absolutely bounded close to one. In \mbox{Section~\ref{section:IsingModelPartitionFunctions}}, we provide an approximation-preserving polynomial-time reduction from the Ising model partition function to this restricted version of the multivariate graph homomorphism partition function. We then use this reduction to establish a deterministic polynomial-time approximation scheme for the Ising model partition function on bounded degree graphs when the interactions and external fields are absolutely bounded sufficiently close to zero. In this regime, we prove that the partition function does not vanish. In \mbox{Section~\ref{section:QuantumSimulation}}, we show how our algorithm can be extended to approximate certain output probability amplitudes of quantum circuits. Finally, we conclude in \mbox{Section~\ref{section:ConclusionAndOutlook}} with some remarks and open problems.

\section{Graph Homomorphism Partition Functions}
\label{section:GraphHomomorphismPartitionFunctions}

A \emph{graph homomorphism} between two graphs $G$ and $H$ is an \mbox{adjacency-preserving} map between the vertex sets, i.e., a map \mbox{$h:V(G) \to V(H)$} such that \mbox{$\{u, v\} \in E(G)$} implies \mbox{$\{h(u), h(v)\} \in E(H)$}. Graph homomorphisms generalise the notion of graph colouring~\cite{hell2004graphs}; for example, a graph homomorphism from a graph $G$ to the complete graph $K_q$ is equivalent to a proper \mbox{$q$-colouring} of $G$. 

Hell and Ne\v{s}et\v{r}il~\cite{hell1990complexity} proved that the problem of deciding if a homomorphism between two graphs $G$ and $H$ exists is \mbox{\textsc{NP}-complete}. Dyer and Greenhill~\cite{dyer2000complexity} showed that the corresponding counting problem is \mbox{\textsc{\#P}-hard}, unless the graph has some special structure; otherwise it is in \textsc{P}. Furthermore, they showed that this problem remains \mbox{\textsc{\#P}-hard} when restricted to graphs of bounded degree. The \emph{graph homomorphism partition function} is defined as follows.
\begin{definition}[Graph homomorphism partition function]
    Let \mbox{$G=(V,E)$} be a graph and let \mbox{$A=(a_{ij})_{m \times m}$} be a $m \times m$ symmetric matrix. Then the graph homomorphism partition function is defined by
    \begin{align}
        \mathrm{Hom}(G;A) := \sum_{\phi:V\to[m]}\prod_{\{u,v\}\in E}a_{\phi(u)\phi(v)}. \notag
    \end{align}
\end{definition}
The graph homomorphism partition function evaluates to many important combinatorial quantities, including counting the number of graph homomorphisms, proper colourings, and independent sets~\cite{barvinok2016combinatorics}.

The complexity of computing graph homomorphism partition functions has been widely studied. Dyer and Greenhill~\cite{dyer2000complexity} showed that computing \mbox{$\mathrm{Hom}(G;A)$} when $A$ is a fixed symmetric binary matrix is either in \textsc{P} or \mbox{\textsc{\#P}-hard}. Moreover, they showed that these hardness results hold even for graphs of maximum degree three. These results were later generalised to non-negative symmetric matrices~\cite{bulatov2005complexity}, real symmetric matrices~\cite{goldberg2010complexity}, and complex symmetric matrices~\cite{cai2010graph}. Furthermore, the tractability criterion for the matrix is decidable in polynomial time.

The graph homomorphism partition function can be generalised by assigning a $m \times m$ symmetric matrix to each edge. The \emph{multivariate graph homomorphism partition function} is defined as follows.
\begin{definition}[Multivariate graph homomorphism partition function]
    Let \mbox{$G=(V,E)$} be a graph with the $m \times m$ symmetric matrices \mbox{$\mathcal{A}=\{(a_{ij}^e)_{m \times m}\}_{e \in E}$} assigned to its edges. Then the multivariate graph homomorphism partition function is defined by
    \begin{align}
        \mathrm{Hom}_\mathrm{M}(G;\mathcal{A}) := \sum_{\phi:V\to[m]}\prod_{\{u,v\}\in E}a_{\phi(u)\phi(v)}^{\{u,v\}}. \notag
    \end{align}
\end{definition}
When the matrices are all equal, it is clear that the multivariate and standard graph homomorphism partition functions are equivalent.

For convenience, let us define the polydisc consisting of all sets of $m \times m$ symmetric matrices with matrix entries absolutely bounded close to one.
\begin{definition}[$\mathcal{D}_{G,m}(\delta)$]
    For a graph \mbox{$G=(V,E)$}, \mbox{$m\in\mathbb{Z}^+$}, and \mbox{$\delta>0$}, we define \mbox{$\mathcal{D}_{G,m}(\delta)$} to be the closed polydisc consisting of all sets of $m \times m$ symmetric matrices \mbox{$\mathcal{A}=\{(a_{ij}^e)_{m \times m}\}_{e \in E}$}, such that \mbox{$\abs{1-a_{ij}^e}\leq\delta$} for all \mbox{$e \in E$} and all \mbox{$i,j \in [m]$}.
\end{definition}

Barvinok and Sober\'on~\cite{barvinok2017computing} gave a \mbox{quasi-polynomial} time algorithm for approximating \mbox{$\mathrm{Hom}_\mathrm{M}(G;\mathcal{A})$} when $G$ is a graph of maximum degree at most $\Delta$ and $\mathcal{A}$ lies in the interior of the closed polydisc \mbox{$\mathcal{D}_{G,m}\left(\delta_\Delta\right)$}. Here, \mbox{$\delta_\Delta>0$} is an absolute constant. The absolute constants come from Barvinok's monograph~\cite{barvinok2016combinatorics}, where a simpler proof was presented with better constants. Patel and Regts~\cite{patel2017deterministic} improved this algorithm to run in polynomial time.
\begin{definition}[$\delta_\Delta$]
    For \mbox{$\Delta\in\mathbb{Z}^+$}, we define the absolute constant $\delta_\Delta$ by
    \begin{align}
        \delta_\Delta := \max_{0<\alpha<\frac{2\pi}{3\Delta}}\left[\sin\left(\frac{\alpha}{2}\right)\cos\left(\frac{\alpha\Delta}{2}\right)\right]. \notag
    \end{align}
\end{definition}
\begin{remark}
    A simple numerical search gives \mbox{$\delta_3=0.18$}, \mbox{$\delta_4=0.13$}, \mbox{$\delta_5=0.11$}, and \mbox{$\delta_6=0.09$}. In general, we have \mbox{$\delta_\Delta=\Omega(1/\Delta)$}.
\end{remark}

We shall consider a restricted version of the multivariate graph homomorphism partition function, in which the sum is restricted to map a subset of vertices to a fixed index.
\begin{definition}[Restricted multivariate graph homomorphism partition function]
    Let \mbox{$G=(V,E)$} be a graph with the $m \times m$ symmetric matrices \mbox{$\mathcal{A}=\{(a_{ij}^e)_{m \times m}\}_{e \in E}$} assigned to its edges. Further let \mbox{$S \subseteq V$} be a subset of $V$ and let \mbox{$k\in[m]$} be an integer. Then the restricted multivariate graph homomorphism partition function is defined by
    \begin{align}
        \mathrm{Hom}_\mathrm{M}(G,S,k;\mathcal{A}) := \sum_{\substack{\phi:V\to[m]\\\phi(s)=k,\forall s \in S}}\prod_{\{u,v\}\in E}a_{\phi(u)\phi(v)}^{\{u,v\}}. \notag
    \end{align}
\end{definition}
The advantage of considering the restricted multivariate graph homomorphism partition function is that, when reduced from the Ising partition function, it will allows us to implement an external magnetic field. This reduction is described in detail in \mbox{Appendix~\ref{section:IsingModelGraphHomomorphismPartitionFunctionReduction}}.

We extend the results of Barvinok and Sober\'on~\cite{barvinok2017computing} and Patel and Regts~\cite{patel2017deterministic} to give a deterministic polynomial-time approximation scheme for the restricted multivariate graph homomorphism partition function. We have the following theorem.  
\begin{theorem}[{restate=[name=restatement]GraphHomomorphismPartitionFunctionBoundedDegreeGraphs}]
    \label{theorem:GraphHomomorphismPartitionFunctionBoundedDegreeGraphs}
    Fix \mbox{$\Delta\in\mathbb{Z}^+$} and \mbox{$0<\delta<\delta_\Delta$}. There is a deterministic polynomial-time approximation scheme for the restricted multivariate graph homomorphism partition function \mbox{$\mathrm{Hom}_\mathrm{M}(G,S,k;\mathcal{A})$} for all graphs \mbox{$G=(V,E)$} of maximum degree at most $\Delta$ and all \mbox{$\mathcal{A}=\{(a_{ij}^e)_{m \times m}\}_{e \in E}$} in the closed polydisc \mbox{$\mathcal{D}_{G,m}\left(\delta\right)$}.
\end{theorem}
We prove \mbox{Theorem~\ref{theorem:GraphHomomorphismPartitionFunctionBoundedDegreeGraphs}} in \mbox{Appendix~\ref{section:GraphHomomorphismPartitionFunctionBoundedDegreeGraphs}}. Our proof requires a result of Barvinok~\cite[Theorem 7.1.4]{barvinok2016combinatorics}, which states that \mbox{$\mathrm{Hom}_\mathrm{M}(G,S,k;\mathcal{A})$} does not vanish on graphs of maximum degree at most $\Delta$ when $\mathcal{A}$ lies in the interior of the closed polydisc \mbox{$\mathcal{D}_{G,m}\left(\delta_\Delta\right)$}.
\begin{lemma}[Barvinok~\cite{barvinok2016combinatorics}]
    \label{lemma:GraphHomomorphismPartitionFunctionZeroFreeRegion}
    Fix \mbox{$\Delta\in\mathbb{Z}^+$}. For any graph \mbox{$G=(V,E)$} of degree at most $\Delta$ and any \mbox{$\mathcal{A}=\{(a_{ij}^e)_{m \times m}\}_{e \in E}$} in the closed polydisc \mbox{$\mathcal{D}_{G,m}\left(\delta_\Delta\right)$}, the restricted multivariate graph homomorphism partition function does not vanish, i.e., \mbox{$\mathrm{Hom}_\mathrm{M}(G,S,k;\mathcal{A})\neq0$} for all \mbox{$S \subseteq V$} and all \mbox{$k\in[m]$}.
\end{lemma}

Our proof also requires the following lemma, which states that we can efficiently compute the constant term and inverse power sums of the roots of \mbox{$\mathrm{Hom}_\mathrm{M}(G,S,k;\mathcal{A}(z))$}.
\begin{lemma}[{restate=[name=restatement]GraphHomomorphismPartitionFunctionComputeConstantTermAndInversePowerSums}]
    \label{lemma:GraphHomomorphismPartitionFunctionComputeConstantTermAndInversePowerSums}
    Fix \mbox{$\Delta\in\mathbb{Z}^+$}, \mbox{$0<\epsilon<1$}, and \mbox{$C>0$}. Let \mbox{$G=(V,E)$} be a graph of maximum degree at most $\Delta$ with the $m \times m$ symmetric matrices \mbox{$\mathcal{A}(z)=\{(1+z(a_{ij}^e-1))_{m \times m}\}_{e \in E}$} assigned to its edges. Further let \mbox{$\{r_i\}_{i=1}^{\abs{E}}$} be the roots of the polynomial \mbox{$P(G,S,k;z):=\mathrm{Hom}_\mathrm{M}(G,S,k;\mathcal{A}(z))$}. Then there is a deterministic \mbox{$(\abs{V}/\epsilon)^{O(1)}$-time} algorithm for computing $P(G,S,k,0)$ and the inverse power sums \mbox{$\left\{\sum_{i=1}^{\abs{E}}r_i^{-j}\right\}_{j=1}^m$} for \mbox{$m=C\log(\abs{V}/\epsilon)$}.
\end{lemma}
We prove \mbox{Lemma~\ref{lemma:GraphHomomorphismPartitionFunctionComputeConstantTermAndInversePowerSums}} in \mbox{Appendix~\ref{section:GraphHomomorphismPartitionFunctionComputeConstantTermAndInversePowerSums}}. For convenience, let us define the closed disc $D$ of radius $\delta$ centred at the origin.
\begin{definition}[$D(\delta)$]
    For \mbox{$\delta>0$}, we define \mbox{$D(\delta)$} to be the closed disc consisting of all complex numbers $z$, such that \mbox{$\abs{z}\leq\delta$}.
\end{definition}

Finally, we require the following lemma, which arises from the error analysis of Barvinok's interpolation method~\cite{barvinok2015computing, barvinok2016computing} (see Barvinok~\cite[Lemma 2.2.1]{barvinok2016combinatorics}). The lemma states that, in order to get a multiplicative approximation to a polynomial inside its zero-free disc, it is sufficient to compute the constant term and inverse power sums of its roots.
\begin{lemma}[{name={Barvinok~\cite{barvinok2015computing, barvinok2016computing, barvinok2016combinatorics}}, restate=[name=restatement]{ApproximatePolynomialZeroFreeRegionInversePowerSums}}]
    \label{lemma:ApproximatePolynomialZeroFreeRegionInversePowerSums}
    Fix \mbox{$0<\epsilon<1$}. Let \mbox{$\{r_i\}_{i=1}^n$} be the roots of the polynomial \mbox{$p(z):=\sum_{k=0}^na_kz^k$}. Suppose that, for some \mbox{$\delta>0$}, the roots of $p$ lie in the exterior of the closed disc $D(\delta)$. Suppose further that we can compute $a_0$ and the inverse power sums \mbox{$\left\{\sum_{i=1}^nr_i^{-j}\right\}_{j=1}^m$} in time $\tau(m)$. Then, for any $t$ in the interior of the closed disc $D(\delta)$, we can compute a multiplicative \mbox{$\epsilon$-approximation} to $p(t)$ in time \mbox{$O\left[\tau\left(\frac{\log(n/\epsilon)}{1-\abs{t}/\delta}\right)\right]$}.
\end{lemma}
We prove \mbox{Lemma~\ref{lemma:ApproximatePolynomialZeroFreeRegionInversePowerSums}} in \mbox{Appendix~\ref{section:ApproximatePolynomialZeroFreeRegionInversePowerSums}}.

\section{Ising Model Partition Functions}
\label{section:IsingModelPartitionFunctions}

The Ising model is described by a graph \mbox{$G=(V,E)$}, with the vertices representing spins and the edges representing interactions between them. A set of edge weights \mbox{$\{\omega_e\}_{e \in E}$} characterise the interactions and a set of vertex weights \mbox{$\{\upsilon_v\}_{v \in V}$} characterise the external fields at each spin. A configuration of the model is an assignment $\sigma$ of each spin to one of two possible states \mbox{$\{-1,+1\}$}. The \emph{Ising model partition function} is defined as follows.
\begin{definition}[Ising model partition function]
    Let \mbox{$G=(V,E)$} be a graph with the weights \mbox{$\Omega=\{\omega_e\}_{e \in E}$} assigned to its edges and the weights \mbox{$\Upsilon=\{\upsilon_v\}_{v \in V}$} assigned to its vertices. Then the Ising model partition function is defined by
    \begin{align}
        \mathrm{Z}_\mathrm{Ising}(G;\Omega,\Upsilon) := \sum_{\sigma\in\{-1,+1\}^V}w_G(\sigma), \notag
    \end{align}
    where
    \begin{align}
        w_G(\sigma) = \exp\left(\sum_{\{u,v\}\in E}\omega_{\{u,v\}}\sigma_u\sigma_v+\sum_{v\in V}\upsilon_v\sigma_v\right). \notag
    \end{align}
\end{definition}
The model is called \emph{ferromagnetic} if \mbox{$\omega_e>0$} for all \mbox{$e \in E$} and \emph{anti-ferromagnetic} if \mbox{$\omega_e<0$} for all \mbox{$e \in E$}. Otherwise, the model is called \emph{non-ferromagnetic}.

A classic result of Jerrum and Sinclair~\cite{jerrum1993polynomial} establishes a fully polynomial-time randomised approximation scheme for the Ising model partition function for all graphs in the ferromagnetic regime with real vertex weights. In contrast, they showed that no such scheme could exists in the anti-ferromagnetic regime unless \mbox{\textsc{RP}=\textsc{NP}}. Furthermore, they showed that exactly computing the Ising model partition function is \mbox{\textsc{\#P}-hard}.

We shall extend the result of \mbox{Theorem~\ref{theorem:GraphHomomorphismPartitionFunctionBoundedDegreeGraphs}} to the Ising model partition function. This is achieved by an approximation-preserving polynomial-time reduction from the Ising model partition function to the restricted multivariate graph homomorphism partition function.
\begin{proposition}[{restate=[name=restatement]IsingModelGraphHomomorphismPartitionFunctionReduction}]
    \label{proposition:IsingModelGraphHomomorphismPartitionFunctionReduction}
    There is an approximation-preserving polynomial-time reduction from the Ising model partition function to the restricted multivariate graph homomorphism partition function.
\end{proposition}

We prove \mbox{Proposition~\ref{proposition:IsingModelGraphHomomorphismPartitionFunctionReduction}} in \mbox{Appendix~\ref{section:IsingModelGraphHomomorphismPartitionFunctionReduction}}. Let us define the following closed polyregion, which arises naturally from applying \mbox{Proposition~\ref{proposition:IsingModelGraphHomomorphismPartitionFunctionReduction}} to \mbox{Theorem~\ref{theorem:GraphHomomorphismPartitionFunctionBoundedDegreeGraphs}}. 
\begin{definition}[$\mathcal{R}_G(\delta)$]
    For a graph \mbox{$G=(V,E)$} and \mbox{$\delta>0$}, we define \mbox{$\mathcal{R}_G(\delta)$} to be the closed polyregion consisting of all sets of weights \mbox{$\Omega=\{\omega_e\}_{e \in E}$} and \mbox{$\Upsilon=\{\upsilon_v\}_{v \in V}$}, such that \mbox{$\abs{1-e^{\pm\omega_e}}\leq\delta$} for all \mbox{$e \in E$} and \mbox{$\abs{1-e^{\pm\upsilon_v}}\leq\delta$} for all \mbox{$v \in V$}.
\end{definition}

We have the following corollary of \mbox{Theorem~\ref{theorem:GraphHomomorphismPartitionFunctionBoundedDegreeGraphs}} and \mbox{Proposition~\ref{proposition:IsingModelGraphHomomorphismPartitionFunctionReduction}}.
\begin{corollary}
    \label{corollary:IsingModelPartitionFunctionBoundedDegreeGraphs}
    Fix \mbox{$\Delta\in\mathbb{Z}^+$} and \mbox{$0<\delta<\delta_{\Delta+1}$}. There is a deterministic polynomial-time approximation scheme for the Ising model partition function \mbox{$\mathrm{Z}_\mathrm{Ising}(G;\Omega,\Upsilon)$} for all graphs \mbox{$G=(V,E)$} of maximum degree at most $\Delta$ and all \mbox{$\Omega=\{\omega_e\}_{e \in E}$} and all \mbox{$\Upsilon=\{\upsilon_v\}_{v \in V}$} in the closed polyregion \mbox{$\mathcal{R}_G\left(\delta\right)$}.
\end{corollary}
\begin{proof}
    The proof follows directly from \mbox{Theorem~\ref{theorem:GraphHomomorphismPartitionFunctionBoundedDegreeGraphs}} and \mbox{Proposition~\ref{proposition:IsingModelGraphHomomorphismPartitionFunctionReduction}}, while noting that the reduction from the Ising model partition to the restricted multivariate graph homomorphism partition function increases the maximum vertex degree by one.
\end{proof}

\begin{remark}
    It is possible to marginally increase the size of the polyregion by applying the \emph{\mbox{$k$-thickening}} technique of Jaeger, Vertigan, and Welsh~\cite{jaeger1990computational}.
\end{remark}

It is important to mention that the bounds of \mbox{Corollary~\ref{corollary:IsingModelPartitionFunctionBoundedDegreeGraphs}} are not sharp in general. To see this, let us compare the results in the anti-ferromagnetic regime with no external field, to those of Sinclair, Srivastava, and Thurley~\cite{sinclair2014approximation}. In this case, \mbox{Corollary~\ref{corollary:IsingModelPartitionFunctionBoundedDegreeGraphs}} tells us that there is a deterministic polynomial-time approximation scheme for the Ising model partition function on graphs of maximum degree at most $\Delta$ when \mbox{$\omega_e>-\log(\delta_\Delta+1)$} for all \mbox{$e \in E$} (noting that in the case of no external field the reduction preserves maximum degree). The results of Sinclair, Srivastava, and Thurley~\cite{sinclair2014approximation} give a deterministic polynomial-time approximation scheme when \mbox{$\Delta\geq3$} and \mbox{$\omega_e>-\frac{1}{2}\log\left(\frac{\Delta}{\Delta-2}\right)$} for all \mbox{$e \in E$}. Hence, the bound of \mbox{Corollary~\ref{corollary:IsingModelPartitionFunctionBoundedDegreeGraphs}} is not sharp. It is an open problem to prove a sharp bound in the complex case. 

We also have the following corollary concerning the location of the complex zeros of the Ising model partition function on bounded degree graphs. 
\begin{corollary}
    \label{corollary:IsingModelPartitionFunctionZeroFreeRegion}
    Fix \mbox{$\Delta\in\mathbb{Z}^+$}. For any graph \mbox{$G=(V,E)$} of degree at most $\Delta$ and any \mbox{$\Omega=\{\omega_e\}_{e \in E}$} and \mbox{$\Upsilon=\{\upsilon_v\}_{v \in V}$} in the closed polyregion \mbox{$\mathcal{R}_G\left(\delta_{\Delta+1}\right)$}, the Ising model partition function does not vanish, i.e., \mbox{$\mathrm{Z}_\mathrm{Ising}(G;\Omega,\Upsilon)\neq0$}.
\end{corollary}
\begin{proof}
    The proof follows directly from \mbox{Lemma~\ref{lemma:GraphHomomorphismPartitionFunctionZeroFreeRegion}} and \mbox{Proposition~\ref{proposition:IsingModelGraphHomomorphismPartitionFunctionReduction}}.
\end{proof}

This may be of independent interest in statistical physics as the possible points of physical phase transitions are exactly the real limit points of such complex zeros~\cite{sokal2005multivariate}.

\section{Quantum Simulation}
\label{section:QuantumSimulation}

Complex-valued Ising model partition functions arise naturally in the output probability amplitudes of quantum circuits~\cite{de2011quantum, iblisdir2014low}. In particular, for the class of commuting quantum circuits, known as \emph{Instantaneous Quantum Polynomial-time} (IQP) circuits~\cite{shepherd2009temporally, shepherd2010binary, fujii2017commuting}. In this section we shall show how the results of \mbox{Corollary~\ref{corollary:IsingModelPartitionFunctionBoundedDegreeGraphs}} allow us to approximate output probability amplitudes of IQP circuits and, more generally, universal quantum circuits. First introduced by Shepherd and Bremner~\cite{shepherd2009temporally}, IQP circuits comprise only gates that are diagonal in the \mbox{Pauli-X} basis. An IQP circuit is described by an \emph{X-program}.
\begin{definition}[X-program]
    An X-program is a pair \mbox{$(P,\theta)$}, where \mbox{$P=(p_{ij})_{m \times n}$} is a binary matrix and \mbox{$\theta\in[-\pi,\pi]$} is a real angle. The matrix $P$ is used to construct a Hamiltonian of $m$ commuting terms acting on $n$ qubits, where each term in the Hamiltonian is a product of Pauli-X operators, 
    \begin{align}
        \mathrm{H}_{(P,\theta)} := -\theta\sum_{i=1}^m\bigotimes_{j=1}^nX_j^{p_{ij}}. \notag
    \end{align}
    Thus, the columns of $P$ correspond to qubits and the rows of $P$ correspond to interactions in the Hamiltonian.
\end{definition}
An X-program induces a probability distribution $\mathcal{P}_{(P,\theta)}$ known as an \emph{IQP distribution}.
\begin{definition}[$\mathcal{P}_{(P,\theta)}$]
    For an X-program \mbox{$(P,\theta)$} with \mbox{$P=(p_{ij})_{m \times n}$}, we define $\mathcal{P}_{(P,\theta)}$ to be the probability distribution over binary strings \mbox{$x\in\{0,1\}^n$}, given by
    \begin{align}
        \textbf{Pr}[x] := \abs{\bra{x}\exp\left(-i\mathrm{H}_{(P,\theta)}\right)\ket{0^n}}^2. \notag
    \end{align}
\end{definition}

We shall consider X-programs that are induced by a weighted graph.
\begin{definition}[Graph-induced X-program]
    For a graph \mbox{$G=(V,E)$} with the weights \mbox{$\left\{\omega_e\in[-\pi,\pi]\right\}_{e \in E}$} assigned to its edges and the weights \mbox{$\left\{\upsilon_v\in[-\pi,\pi]\right\}_{v \in V}$} assigned to its vertices, we define the X-program induced by $G$ to be an X-program $\mathcal{X}_G$ such that
    \begin{align}
        \mathrm{H}_{\mathcal{X}_G} = -\sum_{\{u,v\} \in E}\omega_{\{u,v\}}X_uX_v-\sum_{v \in V}\upsilon_vX_v. \notag
    \end{align}
\end{definition}

It will be convenient for us to define $\psi_{G}$ as a specific probability amplitude induced by a weighted graph $G$.
\begin{definition}[$\psi_{G}$]
    For a graph \mbox{$G=(V,E)$} with the weights \mbox{$\left\{\omega_e\in[-\pi,\pi]\right\}_{e \in E}$} assigned to its edges and the weights \mbox{$\left\{\upsilon_v\in[-\pi,\pi]\right\}_{v \in V}$} assigned to its vertices, we define $\psi_{G}$ to be the probability amplitude given by
    \begin{align}
        \psi_{G} := \bra{0^{\abs{V}}}\exp\left(-i\mathrm{H}_{\mathcal{X}_G}\right)\ket{0^{\abs{V}}}. \notag
    \end{align}
\end{definition}
We note that any X-program can be efficiently represented by a graph-induced X-program~\cite{shepherd2009temporally}. Moreover, X-programs are known to become universal for quantum computation under postselection~\cite{bremner2010classical}. Therefore, any quantum amplitude can be expressed in the form of $\psi_{G}$. The output probability amplitudes of such a graph-induced X-program are proportional to Ising model partition functions with imaginary weights.
\begin{proposition}[{restate=[name=restatement]IQPIsingModelPartitionFunctionRelation}]
    \label{proposition:IQPIsingModelPartitionFunctionRelation}
    Let \mbox{$G=(V,E)$} be a graph with the weights \mbox{$\Omega=\left\{\omega_e\in[-\pi,\pi]\right\}_{e \in E}$} assigned to its edges and the weights \mbox{$\Upsilon=\left\{\upsilon_v\in[-\pi,\pi]\right\}_{v \in V}$} assigned to its vertices, then,
    \begin{align}
        \psi_{G} = \frac{1}{2^{\abs{V}}}\mathrm{Z}_\mathrm{Ising}(G;i\Omega,i\Upsilon). \notag
    \end{align}
\end{proposition}
We prove \mbox{Proposition~\ref{proposition:IQPIsingModelPartitionFunctionRelation}} in \mbox{Appendix~\ref{section:IQPIsingModelPartitionFunctionRelation}}. We now apply \mbox{Corollary~\ref{corollary:IsingModelPartitionFunctionBoundedDegreeGraphs}} to \mbox{Proposition~\ref{proposition:IQPIsingModelPartitionFunctionRelation}} to achieve a deterministic polynomial-time approximation scheme for computing $\psi_{G}$ for all graphs of bounded maximum degree with weights absolutely bounded sufficiently close to zero.
\begin{corollary}
    \label{corollary:IQPBoundedDegreeGraphs}
    Fix \mbox{$\Delta\in\mathbb{Z}^+$} and \mbox{$0<\delta<\delta_{\Delta+1}$}. There is a deterministic polynomial-time approximation scheme for the probability amplitude $\psi_{G}$ for all graphs \mbox{$G=(V,E)$} of maximum degree at most $\Delta$ with the edge weights \mbox{$\left\{\omega_e\in[-\pi,\pi]\right\}_{e \in E}$} satisfying \mbox{$\abs{\omega_e}\leq2\arcsin(\delta/2)$} for all \mbox{$e \in E$} and the vertex weights \mbox{$\left\{\upsilon_v\in[-\pi,\pi]\right\}_{v \in V}$} satisfying \mbox{$\abs{\upsilon_v}\leq2\arcsin(\delta/2)$} for all \mbox{$v \in V$}.
\end{corollary}
\begin{proof}
    It follows from \mbox{Corollary~\ref{corollary:IsingModelPartitionFunctionBoundedDegreeGraphs}} and \mbox{Proposition~\ref{proposition:IQPIsingModelPartitionFunctionRelation}} that we
    have a deterministic polynomial-time approximation scheme for computing $\psi_{G}$ for all graphs of maximum degree at most $\Delta$ with \mbox{$\Omega=\{i\omega_e\}_{e \in E}$} and \mbox{$\Upsilon=\{i\upsilon_v\}_{v \in V}$} in the closed polyregion \mbox{$\mathcal{R}_G\left(\delta\right)$}. For weights in the range \mbox{$[-\pi,\pi]$}, this is achieved when \mbox{$\abs{\omega_e}\leq2\arcsin(\delta/2)$} for all \mbox{$e \in E$} and \mbox{$\abs{\upsilon_v}\leq2\arcsin(\delta/2)$} for all \mbox{$v \in V$}. This completes the proof.
\end{proof}

It is known that approximating $\psi_{G}$ up to a multiplicative factor for bounded degree graphs with arbitrary weights in \mbox{$[-\pi,\pi]$} is \mbox{\textsc{\#P}-hard}~\cite{fujii2017commuting}, and so it seems unlikely that \mbox{Corollary~\ref{corollary:IQPBoundedDegreeGraphs}} can be extended to hold in this case. We note that \mbox{Corollary~\ref{corollary:IQPBoundedDegreeGraphs}} applies to graph-induced X-programs with weights absolutely bounded by a constant that depends only on the maximum degree of the underlying graph. This corresponds to Hamiltonians that have been evolved for up to a constant time and Hamiltonians that exhibit limited interference. Furthermore, \mbox{Corollary~\ref{corollary:IQPBoundedDegreeGraphs}} applies to classes of graphs with treewidth growing as the square root of the number of vertices; for example, square lattices. For classes of graphs with logarithmic treewidth a deterministic polynomial-time algorithm is known~\cite{markov2008simulating}.

\section{Conclusion \& Outlook}
\label{section:ConclusionAndOutlook}

We have established a deterministic polynomial-time approximation scheme for the Ising model partition function with complex parameters on bounded degree graphs when the interactions and external fields are absolutely bounded by a constant depending on the maximum degree of the graph. Furthermore, we have proven that the partition function does not vanish for this class of Ising models. Finally, we have shown how our algorithm can be extended to approximate certain output probability amplitudes of quantum circuits.

There are a number of interesting problems that remain to be solved, the most obvious of which is to sharpen the bounds of \mbox{Corollary~\ref{corollary:IsingModelPartitionFunctionBoundedDegreeGraphs}}. One approach would be to improve \mbox{Lemma~\ref{lemma:GraphHomomorphismPartitionFunctionZeroFreeRegion}}, i.e., prove that the restricted multivariate graph homomorphism partition function does not vanish on a polydisc of a greater radius. It may also be possible to prove sharper bounds for specific graphs of interest. An alternative approach would be to use decay of correlation based arguments~\cite{weitz2006counting, sly2010computational, sinclair2014approximation}. It is an important open problem to understand the relationship between the location of complex zeros, decay of correlations, and the computational complexity of a function. The work of Liu, Sinclair, and Srivastava~\cite{liu2018fisher} makes significant progress towards resolving this problem.

\section*{Acknowledgements}

We thank Gavin Brennen, Jacob Bridgeman, Christopher Chubb, David Gosset, Richard Jozsa, and Hakop Pashayan for helpful discussions. This research was conducted by the ARC Centre of Excellence for Quantum Computation and Communication Technology (CQC2T), project number CE170100012.

\onecolumngrid

\appendix

\section{Proof of Theorem~\ref*{theorem:GraphHomomorphismPartitionFunctionBoundedDegreeGraphs}}
\label{section:GraphHomomorphismPartitionFunctionBoundedDegreeGraphs}

We shall now prove Theorem~\ref{theorem:GraphHomomorphismPartitionFunctionBoundedDegreeGraphs}.

\GraphHomomorphismPartitionFunctionBoundedDegreeGraphs*

\begin{proof}
    Define \mbox{$P(G,S,k;z):=\mathrm{Hom}_\mathrm{M}(G,S,k;\mathcal{A}(z))$}, with \mbox{$\mathcal{A}(z)=\{(1+z(a_{ij}^e-1))_{m \times m}\}_{e \in E}$} and note that \mbox{$\mathcal{A}=\mathcal{A}(1)$}. By Lemma~\ref{lemma:GraphHomomorphismPartitionFunctionZeroFreeRegion}, we have that $P(G,S,k;z)$ does not vanish whenever $\mathcal{A}(z)$ lies in the closed polydisc $\mathcal{D}_{G,m}(\delta_\Delta)$. Since $\mathcal{A}(1)$ lies in the closed polydisc $\mathcal{D}_{G,m}(\delta)$, $P(G,S,k;z)$ does not vanish for all \mbox{$\abs{z}\leq\delta_\Delta/\delta$}. Let $\{r_i\}_{i=1}^{\abs{E}}$ be the roots of $P(G,S,k;z)$. Then, by setting \mbox{$C=(1-\delta/\delta_\Delta)^{-1}$} in Lemma~\ref{lemma:GraphHomomorphismPartitionFunctionComputeConstantTermAndInversePowerSums}, we have that, for any \mbox{$0<\epsilon<1$}, there is a deterministic \mbox{$(\abs{V}/\epsilon)^{O(1)}$-time} algorithm for computing $P(G,S,k;0)$ and the inverse power sums \mbox{$\left\{\sum_{i=1}^{\abs{E}}r_i^{-j}\right\}_{j=1}^m$} for \mbox{$m=(1-\delta/\delta_\Delta)^{-1}\log(\abs{V}/\epsilon)$}. Then, it follows from Lemma~\ref{lemma:ApproximatePolynomialZeroFreeRegionInversePowerSums} that there is a deterministic \mbox{$(\abs{V}/\epsilon)^{O(1)}$-time} algorithm for computing a multiplicative \mbox{$\epsilon$-approximation} to $P(G,S,k;z)$ for all \mbox{$\abs{z}<\delta_\Delta/\delta$}. Since \mbox{$\delta<\delta_\Delta$}, we can take \mbox{$z=1$}. Hence, we have a deterministic polynomial-time algorithm for computing a multiplicative \mbox{$\epsilon$-approximation} to \mbox{$\mathrm{Hom}_\mathrm{M}(G,S,k;\mathcal{A})$}. This completes the proof.
\end{proof}

\section{Proof of Lemma~\ref*{lemma:GraphHomomorphismPartitionFunctionComputeConstantTermAndInversePowerSums}}
\label{section:GraphHomomorphismPartitionFunctionComputeConstantTermAndInversePowerSums}

We shall now prove Lemma~\ref{lemma:GraphHomomorphismPartitionFunctionComputeConstantTermAndInversePowerSums}. Our proof follows from a generalisation of a result due to Patel and Regts~\cite{patel2017deterministic} \mbox{(Lemma~\ref{lemma:ECBIRGCPComputeConstantTermAndInversePowerSums})} and an additional lemma \mbox{(Lemma~\ref{lemma:ECBIRGCPGraphHomomorphismPartitionFunction})}, which we prove in the remainder of the section.
\GraphHomomorphismPartitionFunctionComputeConstantTermAndInversePowerSums*
\begin{proof}
    The proof follows from combining \mbox{Lemma~\ref{lemma:ECBIRGCPComputeConstantTermAndInversePowerSums}} and \mbox{Lemma~\ref{lemma:ECBIRGCPGraphHomomorphismPartitionFunction}}.
\end{proof}

We shall begin with the following definitions.
\begin{definition}[$\mathcal{G}_n$]
    For $n\in\mathbb{Z}^+$, define $\mathcal{G}_n$ to be the collection of all edge-coloured graphs on at most $n$ vertices.
\end{definition}
\begin{definition}[{$G[U]$}]
    For a graph $G$ and a subset of vertices \mbox{$U \subseteq V(G)$}, define $G[U]$ to be the subgraph induced by $U$.
\end{definition}
\begin{definition}[$\mathrm{Ind}_\mathrm{C}(G,H)$]
    For two edge-coloured graphs $G$ and $H$, define \mbox{$\mathrm{Ind}_\mathrm{C}(G,H)$} to be the number of induced subgraphs of $G$ that are edge-colour isomorphic to $H$.
\end{definition}
\begin{definition}[Multiplicative graph polynomial]
    A graph polynomial $P(G;z)$ is said to be multiplicative if \mbox{$P(\varnothing;z)=1$} and \mbox{$P(G \cup H;z)=P(G;z)P(H;z)$} for any two graphs $G$ and $H$.
\end{definition}
\begin{definition}[Edge-coloured bounded induced graph counting polynomial~\cite{patel2017deterministic}]
    Let $P(G;z)$ be a multiplicative graph polynomial defined by \mbox{$P(G;z):=\sum_{n=0}^{d(G)}\alpha_{G,n}z^n$} with \mbox{$P(G;0)=1$}. We say that $P(G;z)$ is an edge-coloured bounded induced graph counting polynomial if there exists constants \mbox{$\mu,\nu\in\mathbb{Z}^+$}, such that (1) the coefficients $\alpha_{G,n}$ satisfy \mbox{$\alpha_{G,n}=\sum_{H\in\mathcal{G}_{\mu n}}\beta_{H,n}\mathrm{Ind}(H, G)$}, for certain $\beta_{H,n}$ and (2) the coefficients $\beta_{H,n}$ can be computed in time \mbox{$O\left(\nu^{\abs{V(H)}}\right)$}.
\end{definition}

Patel and Regts~\cite[Theorem 3.10]{patel2017deterministic} proved that, for any edge-coloured bounded induced graph counting polynomial, there is an efficient algorithm for computing the constant term and inverse power sums of its roots.
\begin{lemma}[Patel and Regts~\cite{patel2017deterministic}]
    \label{lemma:ECBIGCPComputeConstantTermAndInversePowerSums}
    Fix \mbox{$\Delta\in\mathbb{Z}^+$}, \mbox{$0<\epsilon<1$}, and \mbox{$C>0$}. Let \mbox{$G=(V,E)$} be an edge-coloured graph of maximum degree at most $\Delta$. Further let $P(G;z)$ be an edge-coloured bounded induced graph counting polynomial with roots \mbox{$\{r_i\}_{i=1}^{\deg(P)}$}. Then there is a deterministic \mbox{$(\abs{V}/\epsilon)^{O(1)}$-time} algorithm for computing $P(G,0)$ and the inverse power sums \mbox{$\left\{\sum_{i=1}^{\deg(P)}r_i^{-j}\right\}_{j=1}^m$} for \mbox{$m=C\log(\abs{V}/\epsilon)$}.
\end{lemma}

We shall now generalise the result of Patel and Regts~\cite{patel2017deterministic} to the restricted case, that is, where the graph polynomial is restricted to map a subset of vertices to a fixed index. We begin by extending the previous definitions.
\begin{definition}[Restricted graph]
    A restricted graph is a pair \mbox{$(G,S)$}, where \mbox{$G=(V,E)$} is a graph and \mbox{$S \subseteq V$} is a subset of $V$.
\end{definition}
\begin{definition}[$\mathcal{R}_n$]
    For $n\in\mathbb{Z}^+$, define $\mathcal{R}_n$ to be the collection of all edge-coloured restricted graphs on at most $n$ vertices.
\end{definition}
\begin{definition}[Induced restricted subgraph]
    For a restricted graph \mbox{$(G,S)$} and a subset of vertices \mbox{$U \subseteq V(G)$}, the restricted subgraph induced by $U$ is given by \mbox{$(G[U],S \cap U)$}.
\end{definition}
\begin{definition}[Isomorphic restricted graphs]
    Two restricted graphs \mbox{$(G,S)$} and \mbox{$(H,T)$} are said to be isomorphic if and only if there is an isomorphism $\varphi$ from $G$ to $H$ and $T$ is the image of $S$ under $\varphi$.
\end{definition}
\begin{definition}[{$\mathrm{Ind}_\mathrm{C}\left[(G,S),(H,T)\right]$}]
    For two edge-coloured restricted graphs \mbox{$(G,S)$} and \mbox{$(H,T)$}, define \mbox{$\mathrm{Ind}_\mathrm{C}\left[(G,S),(H,T)\right]$} to be the number of induced restricted subgraphs of \mbox{$(G,S)$} that are edge-colour isomorphic to \mbox{$(H,T)$}.
\end{definition}
\begin{definition}[Multiplicative restricted graph polynomial]
    A restricted graph polynomial \mbox{$P(G,S,k;z)$} is said to be multiplicative if \mbox{$P(\varnothing,\varnothing,k;z)=1$} and \mbox{$P(G \cup H,S \cup T,k;z)=P(G,S,k;z)P(H,T,k;z)$} for any two restricted graphs \mbox{$(G,S)$} and \mbox{$(H,T)$} and integer $k\in\mathbb{Z}^+$.
\end{definition}
\begin{definition}[Edge-coloured bounded induced restricted graph counting polynomial]
    Let \mbox{$P(G,S,k;z)$} be a multiplicative restricted graph polynomial defined by \mbox{$P(G,S,k;z):=\sum_{n=0}^{d(G)}\alpha_{G,S,k,n}z^n$} with \mbox{$P(G,S,k;0)=1$}. We say that \mbox{$P(G,S,k;z)$} is an edge-coloured bounded induced restricted graph counting polynomial if there exists constants \mbox{$\mu,\nu\in\mathbb{Z}^+$}, such that (1) the coefficients $\alpha_{G,S,k,n}$ satisfy \mbox{$\alpha_{G,S,k,n}=\sum_{(H,T)\in\mathcal{R}_{\mu n}}\beta_{H,T,k,n}\mathrm{Ind}_\mathrm{C}\left[(G,S),(H,T)\right]$}, for certain $\beta_{H,T,k,n}$ and (2) the coefficients $\beta_{H,T,k,n}$ can be computed in time \mbox{$O\left(\nu^{\abs{V(H)}}\right)$}.
\end{definition}
The restricted version of Lemma~\ref{lemma:ECBIGCPComputeConstantTermAndInversePowerSums} is then obtained by following the proof of Patel and Regts~\cite{patel2017deterministic} with the definitions extended in the natural way. We omit the proof for the sake of brevity.
\begin{lemma}
    \label{lemma:ECBIRGCPComputeConstantTermAndInversePowerSums}
    Fix \mbox{$\Delta\in\mathbb{Z}^+$}, \mbox{$0<\epsilon<1$}, and \mbox{$C>0$}. Let \mbox{$G=(V,E)$} be an edge-coloured graph of maximum degree at most $\Delta$. Further let \mbox{$P(G,S,k;z)$} be an edge-coloured bounded induced restricted graph counting polynomial with roots \mbox{$\{r_i\}_{i=1}^{\deg(P)}$}. Then there is a deterministic \mbox{$(\abs{V}/\epsilon)^{O(1)}$-time} algorithm for computing \mbox{$P(G,S,k,0)$} and the inverse power sums \mbox{$\left\{\sum_{i=1}^{\deg(P)}r_i^{-j}\right\}_{j=1}^m$} for \mbox{$m=C\log(\abs{V}/\epsilon)$}.
\end{lemma}
\begin{lemma}
    \label{lemma:ECBIRGCPGraphHomomorphismPartitionFunction}
    Let \mbox{$G=(V,E)$} be a graph with the $m \times m$ symmetric matrices \mbox{$\mathcal{A}(z)=\{(1+z(a_{ij}^e-1))_{m \times m}\}_{e \in E}$} assigned to its edges and let each edge \mbox{$e \in E$} be assigned a distinct colour. Further let \mbox{$S \subseteq V$} be a subset of $V$ and let \mbox{$k\in[m]$} be an integer. Then, up to an efficiently computable factor, the restricted multivariate graph homomorphism partition function \mbox{$\mathrm{Hom}_\mathrm{M}(G,S,k;\mathcal{A}(z))$} is an edge-coloured bounded induced graph counting polynomial. 
\end{lemma}
\begin{proof}
    Define \mbox{$P(G,S,k;z)$} by
    \begin{align}
        P(G,S,k;z) := m^{-\abs{V \setminus S}}\mathrm{Hom}_\mathrm{M}(G,S,k;\mathcal{A}(z)). \notag
    \end{align}
    Then,
    \begin{align}
        P(G,S,k;z) &= m^{-\abs{V \setminus S}}\sum_{\substack{\phi:V\to[m]\\\phi(s)=k,\forall s \in S}}\prod_{\{u,v\}\in E}\left[1+z\left(a_{\phi(u)\phi(v)}^{\{u,v\}}-1\right)\right] \notag \\
        &= m^{-\abs{V \setminus S}}\sum_{n=0}^{\abs{E}}z^n\sum_{\substack{F \subseteq E\\\abs{F}=n}}\left[\sum_{\substack{\phi:V\to[m]\\\phi(s)=k,\forall s \in S}}\prod_{\{u,v\}\in F}\left(a_{\phi(u)\phi(v)}^{\{u,v\}}-1\right)\right] \notag \\
        &= \sum_{n=0}^{\abs{E}}z^n\sum_{\substack{F \subseteq E\\\abs{F}=n}}\left[m^{-\abs{V(G[F]) \setminus S}}\sum_{\substack{\phi:V(G[F])\to[m]\\\phi(s)=k,\forall s \in 
        (S \cap V(G[F]))}}\prod_{\{u,v\}\in F}\left(a_{\phi(u)\phi(v)}^{\{u,v\}}-1\right)\right], \notag
    \end{align}
    where $G[F]$ is the subgraph of $G$ induced by $F$. Since the number of vertices in $G[F]$ is at most $2\abs{F}$, we can write
    \begin{align}
        P(G,S,k;z) = \sum_{n=0}^{\abs{E}}z^n\sum_{\substack{(H,T)\in\mathcal{R}_{2n}\\\abs{E(H)}=n}}\left[m^{-\abs{V(H) \setminus T}}\sum_{\substack{\phi:V(H)\to[m]\\\phi(t)=k,\forall t \in 
        T}}\prod_{\{u,v\}\in E(H)}\left(a_{\phi(u)\phi(v)}^{\{u,v\}}-1\right)\right]\mathrm{Ind}_\mathrm{C}\left[(G,S),(H,T)\right]. \notag
    \end{align}
    Therefore, we have
    \begin{align}
        P(G,S,k;z) = \sum_{n=0}^{\abs{E}}\alpha_{G,S,k,n}z^n, \notag
    \end{align}
    with
    \begin{align}
        \alpha_{G,S,k,n} = \sum_{\substack{(H,T)\in\mathcal{R}_{2n}\\\abs{E(H)}=n}}\beta_{H,T,k,n}\mathrm{Ind}_\mathrm{C}\left[(G,S),(H,T)\right] \notag
    \end{align}
    and
    \begin{align}
        \beta_{H,T,k,n} = m^{-\abs{V(H) \setminus T}}\sum_{\substack{\phi:V(H)\to[m]\\\phi(t)=k,\forall t \in 
        T}}\prod_{\{u,v\}\in E(H)}\left(a_{\phi(u)\phi(v)}^{\{u,v\}}-1\right). \notag
    \end{align}
    It is clear that \mbox{$P(G,S,k;z)$} is a multiplicative restricted graph polynomial with \mbox{$P(G,S,k;0)=1$}. Furthermore, for any restricted graph \mbox{$(H,T)\in\mathcal{R}_{2n}$}, the coefficients $\beta_{H,T,k,n}$ can be computed in time \mbox{$O\left(m^{\abs{V(H) \setminus S}}\right)$}. Hence, \mbox{$P(G,S,k;z)$} is an edge-coloured bounded induced restricted graph counting polynomial with constants \mbox{$\mu=2$} and \mbox{$\nu=m$}. This completes the proof.
\end{proof}

\section{Proof of Lemma~\ref*{lemma:ApproximatePolynomialZeroFreeRegionInversePowerSums}}
\label{section:ApproximatePolynomialZeroFreeRegionInversePowerSums}

We shall now prove \mbox{Lemma~\ref{lemma:ApproximatePolynomialZeroFreeRegionInversePowerSums}}. The lemma is due to Barvinok~\cite{barvinok2015computing, barvinok2016computing, barvinok2016combinatorics}, however, our proof closely follows that of Patel and Regts~\cite{patel2017deterministic}.
\ApproximatePolynomialZeroFreeRegionInversePowerSums*
\begin{proof}
    Define the function $f(z)$ on the closed disc $D(\delta)$ by
    \begin{align}
        f(z) := \log(p(z)), \notag
    \end{align}
    where the branch of the logarithm is chosen by taking the principal value at $p(0)$. By Taylor's Theorem about the point $t=0$, for each $t$ in the interior of the closed disc $D(\delta)$,
    \begin{align}
        f(t) = \sum_{j=0}^\infty\frac{t^j}{j!}f^{(j)}(0). \notag
    \end{align}
    Define the Taylor expansion truncated at order $m$ by 
    \begin{align}
        T_m(f)(t) := f(0)+\sum_{j=1}^m\frac{t^j}{j!}f^{(j)}(0). \notag
    \end{align}
    Now, let us write $p(z)$ in terms of its roots. By the Factor Theorem,
    \begin{align}
        p(z) = a_n\prod_{i=1}^n(z-r_i). \notag
    \end{align}
    Then,
    \begin{align}
        f(z) = \log(a_n)+\sum_{i=1}^n\log(z-r_i). \notag
    \end{align}
    Therefore,
    \begin{align}
        f^{(j)}(0) = -(j-1)!\sum_{i=1}^nr_i^{-j}. \notag
    \end{align}
    Let $s_j$ be the $j^{\mathrm{th}}$ inverse power sum given by
    \begin{align}
        s_j := \sum_{i=1}^{n}r_i^{-j}. \notag
    \end{align}
    Then, by noting that \mbox{$f(0)=\log(a_0)$},
    \begin{align}
        T_m(f)(t) = \log(a_0)-\sum_{j=1}^m\frac{s_jt^j}{j}. \notag
    \end{align}
    We shall now show that, for any \mbox{$0<\epsilon<1$}, the Taylor expansion truncated at order \mbox{$m=O(\log(n/\epsilon))$} gives an additive \mbox{$\epsilon$-approximation} to $f(t)$.
    \begin{align}
        \abs{f(t)-T_m(f)(t)} &\leq \abs{\sum_{j=m+1}^\infty\frac{s_jt^j}{j}} \notag \\
        &\leq \frac{1}{m+1}\sum_{j=m+1}^\infty\abs{s_jt^j}. \notag
    \end{align}
    Since the roots \mbox{$\{r_i\}_{i=1}^n$} lie in the exterior of the closed disc $D(\delta)$, we have \mbox{$\abs{s_j}<n/\delta^j$}. Therefore, 
    \begin{align}
        \abs{f(t)-T_m(f)(t)} \leq \frac{n}{m+1}\sum_{j=m+1}^\infty\left(\frac{\abs{t}}{\delta}\right)^j. \notag
    \end{align}
    Since \mbox{$\abs{t}<\delta$}, by the geometric series formula,
    \begin{align}
        \abs{f(t)-T_m(f)(t)} \leq \frac{n(\abs{t}/\delta)^{m+1}}{(m+1)(1-\abs{t}/\delta)}. \notag
    \end{align}
    Taking \mbox{$m=(1-\abs{t}/\delta)^{-1}\log(n/\epsilon)$}, it follows that
    \begin{align}
        \abs{f(t)-T_m(f)(t)} \leq \epsilon. \notag
    \end{align}
    We shall now show that the truncated Taylor expansion is a multiplicative \mbox{$\epsilon$-approximation} to $p(t)$. For the norm, we have
    \begin{align}
        \abs{e^{T_m(f)(t)-f(t)}} &\leq e^{\abs{T_m(f)(t)-f(t)}} \notag \\
        &\leq e^\epsilon, \notag
    \end{align}
    and
    \begin{align}
        \abs{e^{f(t)-T_m(f)(t)}} \leq e^\epsilon. \notag
    \end{align}
    Now, for the argument,
    \begin{align}
        \abs{\operatorname{Arg}\left(e^{T_m(f)(t)-f(t)}\right)} &= \abs{\operatorname{Im}\left[\log\left(e^{f(t)-T_m(f)(t)}\right)\right]} \notag \\
        &\leq \abs{\log\left(e^{f(t)-T_m(f)(t)}\right)} \notag \\
        &\leq \epsilon. \notag
    \end{align}
    This completes the proof.
\end{proof}

\section{Proof of Proposition~\ref*{proposition:IsingModelGraphHomomorphismPartitionFunctionReduction}}
\label{section:IsingModelGraphHomomorphismPartitionFunctionReduction}

We shall now prove \mbox{Proposition~\ref{proposition:IsingModelGraphHomomorphismPartitionFunctionReduction}}.
\IsingModelGraphHomomorphismPartitionFunctionReduction*
\begin{proof}
    Let \mbox{$G=(V,E)$} be a graph with the $2 \times 2$ symmetric matrices \mbox{$\mathcal{A}=\{(a_{ij}^e)_{2 \times 2}\}_{e \in E}$} assigned to its edges. Let us construct a new graph $G'$ from $G$ by the following vertex gadget. For every vertex \mbox{$v \in V$}, add a new vertex $s_v$ and an edge \mbox{$e_v=\{v, s_v\}$} with a $2 \times 2$ symmetric matrix \mbox{$(b_{ij}^{e_v})_{2 \times 2}$} assigned to it. Let \mbox{$S=\{s_v\}_{v \in V}$}, and let \mbox{$\mathcal{B}=\{(b_{ij}^{e_v})_{2 \times 2}\}_{v \in V}$}. Then,
    \begin{align}
        \mathrm{Hom}_\mathrm{M}(G',S,2;\mathcal{A}\cup\mathcal{B}) &= \sum_{\substack{\phi:V(G')\to[2]\\\phi(s)=2,\forall s \in S}}\prod_{\{u,v\} \in E(G)}a_{\phi(u)\phi(v)}^{\{u,v\}}\prod_{v \in V(G)}b_{\phi(v)\phi(s_v)}^{e_v} \notag \\
        &= \sum_{\phi:V(G)\to[2]}\prod_{\{u,v\}\in E(G)}a_{\phi(u)\phi(v)}^{\{u,v\}}\prod_{v \in V(G)}b_{\phi(v)(2)}^{e_v}. \notag
    \end{align}
    Taking \mbox{$a_{ij}^e=\exp\left[\omega_e(2i-3)(2j-3)\right]$} and \mbox{$b_{ij}^{e_v}=\exp\left[\upsilon_v(2i-3)(2j-3)\right]$},
    \begin{align}
        \mathrm{Hom}_\mathrm{M}(G',S,2;\mathcal{A}\cup\mathcal{B}) &= \sum_{\phi:V(G)\to\{-1,+1\}}\exp\left(\sum_{\{u,v\} \in E(G)}\omega_{\{u,v\}}\phi(u)\phi(v)+\sum_{v \in V(G)}\upsilon_v\phi(v)\right) \notag \\
        &= \sum_{\sigma\in\{-1,+1\}^V}\exp\left(\sum_{\{u,v\} \in E(G)}\omega_{\{u,v\}}\sigma_u\sigma_v+\sum_{v \in V(G)}\upsilon_v\sigma_v\right) \notag \\
        &= \mathrm{Z}_\mathrm{Ising}(G;\Omega,\Upsilon), \notag
    \end{align}
    where \mbox{$\Omega=\{\omega_e\}_{e \in E}$} and \mbox{$\Upsilon=\{\upsilon_v\}_{v \in V}$}. Hence, we have an approximation-preserving polynomial-time reduction from the Ising model partition function to the restricted multivariate graph homomorphism partition function. This completes the proof.
\end{proof}

\section{Proof of Proposition~\ref*{proposition:IQPIsingModelPartitionFunctionRelation}}
\label{section:IQPIsingModelPartitionFunctionRelation}

We shall now prove \mbox{Proposition~\ref{proposition:IQPIsingModelPartitionFunctionRelation}}.
\IQPIsingModelPartitionFunctionRelation*
\begin{proof}
    By definition,
    \begin{align}
        \psi_{G} &= \bra{0^{\abs{V}}}\exp\left(i\sum_{\{u,v\} \in E}\omega_{\{u,v\}}X_uX_v+i\sum_{v \in V}\upsilon_vX_v\right)\ket{0^{\abs{V}}} \notag \\
        &= \bra{+^{\abs{V}}}\exp\left(i\sum_{\{u,v\} \in E}\omega_{\{u,v\}}Z_uZ_v+i\sum_{v \in V}\upsilon_vZ_v\right)\ket{+^{\abs{V}}} \notag \\
        &= \frac{1}{2^{\abs{V}}}\sum_{x,y\in\{0,1\}^V}\bra{y}\exp\left(i\sum_{\{u,v\} \in E}\omega_{\{u,v\}}Z_uZ_v+i\sum_{v \in V}\upsilon_vZ_v\right)\ket{x} \notag \\
        &= \frac{1}{2^{\abs{V}}}\sum_{x\in\{0,1\}^V}\exp\left(i\sum_{\{u,v\} \in E}\omega_{\{u,v\}}(-1)^{x_u \oplus x_v}+i\sum_{v \in V}\upsilon_v(-1)^{x_v}\right) \notag \\
        &= \frac{1}{2^{\abs{V}}}\sum_{z\in\{-1,+1\}^V}\exp\left(i\sum_{\{u,v\} \in E}\omega_{\{u,v\}}z_uz_v+i\sum_{v \in V}\upsilon_vz_v\right) \notag \\
        &= \frac{1}{2^{\abs{V}}}\mathrm{Z}_\mathrm{Ising}(G;i\Omega,i\Upsilon). \notag
    \end{align}
    This completes the proof.
\end{proof}

\twocolumngrid

\bibliography{bibliography}

\end{document}